%% file: mainTCS.tex
\documentclass[preprint,12pt]{elsarticle}

\usepackage{amssymb}

\usepackage{lineno}
\usepackage{graphicx}
\usepackage{xcolor}
\usepackage{tikz}
\usetikzlibrary{calc, shapes}
\usepackage{caption}
\usepackage{subcaption}
\usepackage{amssymb,amsthm,latexsym}
\usepackage[left=25mm,right=25mm,top=30mm,bottom=30mm]{geometry}
\usepackage{amsmath,graphicx}

\newcommand{\centroidT}{\ensuremath{\mathcal{T_C}}}
\newcommand{\rle}{\ensuremath{{\sf rle}}}

\newcommand{\stcentroidT}{\ensuremath{\mathcal{ST_C}}}
\newcommand{\rS}[1]{\ensuremath{\hat{S}_{#1}}}
\newcommand{\lz}[1]{\ensuremath{{\sf Z}_{#1}}}
\newcommand{\sufft}{\ensuremath{\mathcal{ST}}}
\newcommand{\locus}{\sf locus}
\newcommand{\child}{\sf children}

\newtheorem{theorem}{Theorem}

\newtheorem{lemma}[theorem]{Lemma}

\journal{Theoretical Computer Science}

\begin{document}

\begin{frontmatter}



\title{Adaptive Learning of Compressible Strings\tnoteref{support}}
\tnotetext[mytitlenote]{Supported by MIUR PRIN 2017 Project 2017K7XPAN and Universit\`a di Pisa Project PRA\_2020-2021\_26.}


\author[inst1]{Gabriele Fici}

\affiliation[inst1]{organization={Dipartimento di Matematica e Informatica, Universit\`a degli Studi di Palermo},
            city={Palermo},
            country={Italy},
            email={, gabriele.fici@unipa.it}}

\author[inst2]{Nicola Prezza}

\affiliation[inst2]{organization={Dipartimento di Scienze Ambientali, Informatica e Statistica, Universit\`a Ca' Foscari},
            city={Venezia},
            country={Italy},
            email={, nicola.prezza@unive.it}}

\author[inst3]{Rossano Venturini}

\affiliation[inst3]{organization={Dipartimento di Informatica, Universit\`a di Pisa},
            city={Pisa},
            country={Italy},
            email={, rossano.venturini@unipi.it}}

\begin{abstract}
Suppose an oracle knows a string $S$ that is unknown to us and that we want to determine. The oracle can answer queries of the form ``Is $s$ a substring of $S$?''. 
In 1995, Skiena and Sundaram showed that, in the worst case,  any algorithm needs to ask the oracle $\sigma n/4 -O(n)$ queries in order to be able to reconstruct the hidden string, where $\sigma$ is the size of the  alphabet of $S$ and $n$ its length, and gave an algorithm that spends $(\sigma-1)n+O(\sigma \sqrt{n})$ queries to reconstruct $S$. 
The main contribution of our paper is to improve the above upper-bound in the context where the string is compressible.
We first present a universal algorithm that, given a (computable) compressor that compresses the string to $\tau$ bits, performs $q=O(\tau)$ substring queries;
this algorithm, however, runs in exponential time. 
For this reason, the second part of the paper focuses on more time-efficient algorithms whose number of queries is bounded by  specific compressibility measures. 
We first show that any string of length $n$ over an integer alphabet of size $\sigma$ with \rle\ runs can be reconstructed with $q=O(\rle (\sigma + \log \frac{n}{\rle}))$ substring queries in linear time and space. We then present an algorithm that spends $q \in O(\sigma g\log n)$ substring queries and runs in $O(n(\log n + \log \sigma)+ q)$ time using linear space, where $g$ is the size of a smallest straight-line program generating the string.
\end{abstract}




\end{frontmatter}


\section{Introduction}
\label{sec:intro}

String reconstruction (or learning) from substrings queries is a well-established  problem that has natural applications in many areas, including bioinformatics, data compression, security, etc. (see, for example,~\cite{JiangLi94,MS95,DBLP:journals/jcb/SkienaS95,Tsur05}). 

In a more general setting, one is interested in understanding whether and how it is possible to reconstruct an unknown target string $S$ from some piece of information about $S$. This information can be, for example, a collection of substrings (e.g., the classical NP-hard Shortest Superstring Problem), or substring compositions (\cite{DBLP:journals/siamdm/AcharyaDMOP15}), or subwords (\cite{Dress-Erdos}), of $S$. Furthermore, the problem can be viewed from different angles, e.g., combinatorial, computational, algorithmic, information theoretical. 

In this paper, we deal with the problem of reconstructing a string from information about its substrings.
Apart from the classical static model for the reconstruction (exact or with uncertainty), many different models have been introduced in the literature for the string reconstruction problem, including the one we consider in this paper, and which has been  presented in 1995 by Skiena and Sundaram~\cite{DBLP:journals/jcb/SkienaS95}. In this model, one can ask an oracle, which knows the target string $S$, queries of the form ``Is $s$ a substring of $S$?'' and is interested in designing an \emph{adaptive algorithm} minimizing the number of such queries. In this setting, \emph{adaptive} means that the algorithm may reuse the information resulting from previous queries in order to decide which queries to ask next.  

It is worth mentioning that, with the same model, other query complexities have been investigated very recently by Amir et al.~\cite{DBLP:conf/spire/AfsharAGM20}. 

A trivial information-theoretic argument implies a worst-case lower bound of $n\log \sigma$ queries, where $\sigma$ is the size of the alphabet of $S$. Skiena and Sundaram~\cite{DBLP:journals/jcb/SkienaS95} improved this bound and showed that $\sigma n/4-O(n)$ queries are necessary to reconstruct $S$ in the worst case. This remains true even if the oracle returns for each substring query the number of its occurrences in $S$. In the same paper, they gave an algorithm for the reconstruction which spends at most $(\sigma-1)n+O(\sigma \sqrt{n})$ queries, thus asymptotically matching the  lower bound. They also gave an algorithm that spends at most $(\sigma-1)n +2\log n + O(\sigma)$ queries if the length $n$ of $S$ is known. Iwama et al.~gave an algorithm for binary strings that spends $n+O(1)$ queries \emph{on average}~\cite{DBLP:journals/corr/abs-1808-00674}. Amir et al.~\cite{DBLP:conf/spire/AfsharAGM20} recently proved that if the string has period $p>0$, then it can be reconstructed using $O(\sigma p + \lg n)$ substring queries, even if both $n$ and $p$ are unknown.

We stress out that these bounds hold in the \emph{adaptive} case: answers to previous queries can be used to decide the next query. As shown by Skiena et al. \cite{MS95}, the non-adaptive model is much harder: if the algorithm has to reconstruct all substrings (of the unknown string) of length $k \leq n$ after a pre-determined batch of queries, then $\sigma^{k/2}/k$ queries are needed to solve the problem. Tsur \cite{Tsur05} explored this model more in detail, providing bounds as a function  of the fraction $(1-\epsilon)$ (with $0\leq \epsilon \leq 1$) of substrings  of length $k$ that have to be reconstructed: in this case, $\Omega(\epsilon^{-1/2}k^2)$ non-adaptive queries are sufficient and necessary.

\subsection{A novel approach to the problem}

While the aforementioned papers tackled the problem in the worst-case, the minimum number of queries needed to reconstruct a string may be significantly smaller than the worst-case on particular instances. For example, consider a string of the form $a^n$, where $a$ is a single letter. An algorithm could first try to find out if the string is of this form by issuing $O(\log n) + 2\sigma$ queries and, only if the check fails, proceed with Skiena and Sundaram's algorithm~\cite{DBLP:journals/jcb/SkienaS95}. Observe that the resulting algorithm optimizes for a particular class of \emph{highly-compressible} strings. In fact, in this paper we show that this reasoning continues to hold for \emph{any compressor}. Our first result is a universal algorithm that, given as input a computable compressor $\mathcal C$, performs the reconstruction asking a number of queries that is proportional to the bit-size $|\mathcal C(S)|$ of the string compressed by $\mathcal C$.
We complement this result by showing that any deterministic adaptive algorithm for reconstructing a string yields a string compressor. 
Together, these results imply the equivalence between the string reconstruction and compression problems. 

Motivated by the fact that our universal algorithm performs an exponential number of calls to $\mathcal C$, we then focus on optimizing the running time and the space usage for commonly used compressors, including run-length encoding, Lempel-Ziv factorization and context-free grammars. 

In measuring the efficiency of an algorithm, we 
assume that any query can be submitted to the oracle in constant time and space regardless of the length of the queried substring. The reason for this assumption is that the implementation of the oracle 
strongly depends on the application. For example, if the application admits a collaborative oracle, there are several possible approaches to achieve constant query time, e.g., using hashing. 
Moreover, one could also assume that the oracle knows the reconstruction strategy and therefore it could run the reconstruction algorithm itself, that is, we do not even need to transmit the next substring query because the oracle already knows the next query it has to answer. 

\subsection{Preliminary definitions}

Let $S$ be a binary string. 
A \emph{compressor} is an injective computable function $\mathcal C : \{0,1\}^+ \rightarrow \{0,1\}^+$ that converts any $S \in \{0,1\}^+$ into a reversible representation $\mathcal C(S)$ of size $|\mathcal C(S)|$ bits. We require also the inverse function $\mathcal C^{-1}$ (i.e. the function such that $\mathcal C^{-1}(\mathcal C(S))=S$) to be computable; this function is the \emph{decompressor} associated with $\mathcal C$. Informally speaking, a function $\mathcal C$ qualifies as a good compressor if $|\mathcal C(S)| \ll |S|$ on particular string families (for example, repetitive strings), and $|\mathcal C(S)| \in O(|S|)$ for all other strings outside this family.

A popular compressor is the LZ77 factorization of $S$. The Lempel-Ziv 1977 (LZ77) algorithm~\cite{LZ77} parses a string $S$ into a sequence of $z$ phrases, where each new phrase is either a fresh character or the longest string that also occurs starting from a position strictly smaller than the phrase start position.
The bit-size of the LZ77 factorization of $S$ is  $\Theta(z\log n)$. 
For example, the LZ77 factorization of the string $abbabba$ is $a|b|b|abba$. In this example, the string is factored into $z = 4$ phrases. A more restricted version of LZ77 does not allow overlaps between a phrase and its source. We denote this version as \emph{LZ77 without overlaps} and with $z_{no}$ the number of generated phrases. Between those two measures, it holds $z_{no} \in O(z \log n)$ \cite{navarro2020indexing}. Clearly, also $z \leq z_{no}$ always holds. This version factorizes the above string as $a|b|b|abb|a$, with $z_{no} = 5$. 
    
Another common measure of compressibility is the number $\rle$ of equal-letter runs in $S$, that is, the number of maximal unary substrings of $S$. This measure is not as strong as $z$; in fact, it is easy to see that $z_{no} \leq \rle$.

In this work we also consider the size (number of nonterminals) $g$ of the smallest straight-line program (SLP) producing (only) $S$. SLPs are particular cases of acyclic context-free grammars composed of rules of the kind $A \rightarrow BC$, where $A$ is a nonterminal and $B,C$ are either nonterminals or terminals.
 
The known relations between $g$ and $z_{no}$ are $g \in O(z_{no}\log(n/z_{no}))$ and $z_{no} \leq g$. 
See Navarro's recent survey~\cite{navarro2020indexing} for more details on these and several other relations between string complexity measures. 



 

\section{Universal string reconstruction}\label{sec:universal}


In this section we present an algorithm that, given a compressor $C$, reconstructs any string $S$ with $O(|C(S)|)$ queries to the oracle. We furthermore prove a dual result: any reconstruction algorithm performing $\chi(S)$ queries yields a compression algorithm (with associated decompressor) that compresses the string to $\chi(S)$ bits. These findings show that the string reconstruction problem essentially coincides with the string compression problem. For simplicity, in this section we restrict our attention to binary alphabets only.


We start with a lemma of Skiena and Sundaram \cite{DBLP:journals/jcb/SkienaS95} stating that any set $M$ of binary strings admits a string that is a substring of a constant fraction of the members of $M$. 
Letting $M$ be a set of strings, we let $M(S)$ denote the subset of $M$ whose elements contain $S$ as a substring. 

\begin{lemma}\label{lem:1/5}
	(\cite[Lem. 12]{DBLP:journals/jcb/SkienaS95}) Let $M \subseteq \{0,1\}^n $ be a set of binary strings, each of length $n$. 
	Then, there exists a string $S$ such that $ \frac{1}{5} |M| \leq |M(S)| \leq \frac{4}{5} |M|$.
\end{lemma}

Lemma \ref{lem:1/5} can be turned into a universal algorithm for determining the substring queries to be asked to the oracle as a function of any given compressor. 

\begin{lemma}\label{lem:upper}
	Let $\mathcal C$ be a compressor, and let $S \in \{0,1\}^n$ be an unknown binary string of known length $n$. 
	Then, there is an algorithm that reconstructs $S$ using $O(|\mathcal C(S)|)$ substring queries.
\end{lemma}
\begin{proof}
    Let $M_k = \{S \in \{0,1\}^n\ :\ |\mathcal C(S)| \leq k\}$ be the set of strings of length $n$ compressed to at most $k$ bits by $\mathcal C$. Note that $|M_k| \leq 2^{k+1}-2$, since $\mathcal C$ is injective and there are no more than $2^{k+1}-2$ binary strings (compressed representations) of length at most $k$. 
    Assuming we know the value of $\tau = |\mathcal C(S)|$ (later we remove this assumption), it is easy to design an (exponential-time) algorithm that builds $M_\tau$: simply apply $\mathcal C$ to all strings of length $n$, keeping only those such that $|\mathcal C(S)| \leq \tau$. 
    By definition of $\tau$, note that $S\in M_\tau$.
    Then, by applying  recursively Lemma \ref{lem:1/5} starting from the set $M_\tau$, we end up selecting $S$ from this set.
	Each recursive iteration yields a string that we use to perform a substring query on $S$, thereby reducing the number of candidates by a factor 4/5 in the worst case. 
	After $O(\log |M_\tau|) = O(\tau)$ iterations (i.e., substring queries on $S$), we discover which element of $M_\tau$ corresponds to $S$. 
	To conclude, we can remove the assumption that we know $\tau$. To achieve this goal it is sufficient to run an exponential search on the above strategy, i.e., run it on $M_{\tau'}$ for $\tau' = 1, 2, 4, \dots, 2^{\lceil \log_2 \tau\rceil}$. The last iteration will reveal $S$, after a total of $O(\tau)$ substring queries on $S$. 
\end{proof}

We finally prove the following lemma.

\begin{lemma}
Let $A$ be a deterministic adaptive algorithm that reconstructs any string $S$ by asking $\chi(S)$ queries to the oracle, for some (computable) function $\chi(S)$. Then, there exists a compressor $\mathcal C$ (and an associated decompressor $\mathcal C^{-1}$) such that $|\mathcal C(S)| = \chi(S)$.
\end{lemma}
\begin{proof}
It is straightforward to turn $A$ into a compressor $\mathcal C$: the compressed representation $\mathcal C(S)$ of $S$ is the binary string of length $\chi(S)$ formed by the $\chi(S)$ answers received by the oracle while reconstructing $S$. The $\chi(S)$ answers can be computed by any pattern matching algorithm testing membership of the substrings queried by $A$ in the substring closure of $S$. Similarly, $A$ itself can be turned into a decompressor: by definition of $A$, the $\chi(S)$ answers of the oracle (i.e. the compressed file representation) are sufficient to reconstruct $A$.
\end{proof}

While the above results establish an asymptotic equivalence between the string reconstruction and compression problems, they do not yield time-efficient algorithms for reconstructing a string in time proportional to its compressed size. In the next section we tackle this problem by focusing on particular string compressors. 

\section{Feasible algorithms for the reconstruction}\label{sec:algo}

Let $S$ be a string of length $n$ over an integer alphabet $\Sigma=[1, \ldots, \sigma]$.
A trivial algorithm for reconstructing $S$ with $\sigma(n + 1)$ substring queries is the following~\cite{DBLP:journals/jcb/SkienaS95}: We make queries of single character substrings, so that after at most $\sigma$ queries a new character of $S$ is determined. Let $s$ be a known substring of $S$.
In general, we can increase the length of this known substring by one character by querying on the strings $s\sigma_i$, for every character $\sigma_i$. At least one of these queries must be a substring of $S$, unless $s$ is a suffix of $S$ that has no other occurrences in $S$.
When $s$ can no longer be extended to the right, i.e., $s$ is a suffix of $S$ not appearing elsewhere in $S$, we can continue the process by prepending characters to the known substring $s$, until it can no longer be extended to the left, and the string $S$ is then reconstructed.

This algorithm is optimal up to constant factors due to the following lower bound~\cite{DBLP:journals/jcb/SkienaS95}.

\begin{theorem}\label{thm:lb}(\cite[Thm. 8]{DBLP:journals/jcb/SkienaS95})
In the worst case, $\frac{\sigma n}{4}-O(n)$ substring queries are necessary to reconstruct a string of length $n$.
\end{theorem}

In the rest of this section, we will provide algorithms for reconstructing the string $S$ whose efficiency is measured towards commonly used measures of compression for strings. 

Let us first show an easy result for the size \rle\ of the run-length encoding of $S$,
i.e., \rle\ is the number of runs (maximal repetitions of the same character) in $S$.
We show that $S$ can be reconstructed with $O(\rle (\sigma + \log \frac{n}{\rle}))$ queries. The reconstruction is done in \rle\ steps.
Let $\rS{i-1}$ be the substring reconstructed so far. In the $i$th step, we first identify the 
character $c$ that follows $\rS{i-1}$ in $S$. This is done by querying $\rS{i-1}\cdot c$
for any $c\in \Sigma$. Once we know $c$, we need to identify the length of the run of $c$, i.e., the maximal value $r_i$ such that $\rS{i-1}\cdot c^{r_i}$ is a substring of $S$.
This can be done with an exponential search on $r_i$, which takes $\Theta(\log r_i)$ queries.
When at some step $j$, $\rS{j}$ cannot be extended further, we continue the process by prepending (runs of) characters to the known substring.

The overall number of queries is $q=O(\sum_{i=1} ^{\rle} (\sigma + \log r_i))$. 
This is in $O(\rle (\sigma + \log \frac{n}{\rle}))$ queries because the sum of the terms $\log r_i$ is
maximized when every $r_i$ is in  $\Theta(
\frac{n}{\rle})$.

\begin{theorem}
Any string of length $n$ with \rle\ runs can be reconstructed with $q=O(\rle (\sigma + \log \frac{n}{\rle}))$ substring queries in $O(q)$ time and $O(\rle)$ space.
\end{theorem}

Note that, accordingly to Theorem~\ref{thm:lb}, this result is optimal up to a constant factor for sufficiently large $\sigma$. 

Our next aim is to give algorithms whose complexity grows as a function of the size of the LZ77 parsing of $S$. We let $z_{no}$ denote the number of phrases of the LZ77 parsing when the parse does not allow overlapping phrases (both settings are commonly considered in the literature).

We can use Theorem~\ref{thm:lb} to prove a lower bound on $\chi(S)$ in terms of $z_{no}$. 

\begin{theorem}\label{thm:lblz}
In the worst case,  $\Omega(\sigma z_{no} {\log_\sigma n})$ substring queries are necessary to reconstruct a string of length $n$.
\end{theorem}
\begin{proof}
It is well known that for any string $z_{no} = O(\frac{n}{\log_\sigma n})$.
The theorem follows by combining this fact with Theorem~\ref{thm:lb}.
\end{proof}

We are not required to know the length $n$ of $S$, but we assume to know its alphabet $\Sigma=[1, \ldots, \sigma]$. If instead also the alphabet is unknown, we need $O(\log \sigma)$ queries to identify the largest character in $S$. This is done by performing 
an exponential search to identify the largest character occurring in $S$. Notice that this 
is correct only if all the characters in $\Sigma$ occur in $S$ (in particular, $\sigma \leq n$), which we assume as hypothesis.

Our goal is to prove the following theorem.

\begin{theorem}\label{thm:lz}
Let $S$ be a string of length $n$ over the alphabet $\Sigma=[1,\ldots, \sigma]$. 
There exists an algorithm that reconstructs $S$ with $ q = O(\sigma z_{no}\log (n/z_{no})\log n)$ substring queries 
to the oracle.  The algorithm runs in $O(n(\log n + \log \sigma)+ q)$ time using linear space.
\end{theorem}

Note that this result is optimal up to a factor $O(\log \frac{n}{z_{no}} \log \sigma)$  by Theorem~\ref{thm:lblz}.

In the next subsection we review a technique to solve pattern matching queries 
on a text which exploits the centroid decomposition of the suffix tree of a string. 
This will allow us to give an efficient algorithm for reconstructing the suffix tree of $S$, from which $S$ is therefore determined.

\paragraph*{Pattern matching with the centroid decomposition}
 
This technique has been introduced by Naor~\cite{naor91} and has found applications, 
for example, in designing cache-oblivious string B-trees~\cite{BenderFK06,FerraginaV16} or randomized pattern matching \cite{amir92} on a dictionary of strings.

The {\em centroid decomposition} of a tree $\mathcal T$ (also known as {\em separator decomposition}) is a popular and powerful technique to obtain a tree \centroidT\ of logarithmic height.
The decomposition is based on a theorem proved by Jordan in 1869 \cite{Jordan1869}. 

\begin{lemma}
Any tree $\mathcal T$ of $n$ nodes has at least a node, called centroid, whose removal leaves connected components of size at most $n/2$.
\end{lemma}

The centroid decomposition is defined recursively.  Given $\mathcal T$, we identify a centroid node $u$, which is chosen to be the root of the new rooted tree \centroidT . Then, we remove $u$ from $\mathcal T$ and recurse on each connected component to get $u$'s subtrees in \centroidT . The children of $u$ in \centroidT\ are the roots
of the centroid decompositions of these components. Let us use $\child_{\centroidT}(u)$ to denote the set of children of $u$ in \centroidT .
As we have a (possibly empty) component for $u$'s parent and children in $\mathcal T$, the outdegree of $u$ in \centroidT, i.e., $|\child_{\centroidT}(u)|$,  is at most the outdegree of $u$ in $\mathcal T$ plus one.
The resulting decomposition is a new tree \centroidT\ on the same nodes, whose height is $\Theta(\log n)$.

A folklore algorithm computes the centroid decomposition in $\Theta(n\log n)$ time as follows.
We first observe that a centroid node of $\mathcal T$ can be easily identified in linear time.
Indeed, we can arbitrary choose a root in $\mathcal T$ and visit the tree to compute the size
of each subtree. Then, we start from the root and move to the largest subtree until we reach a
node whose subtrees have size at most $n/2$.  This node is a centroid of the tree.
The centroid decomposition is computed by repeating the above algorithm recursively in each component. It easily follows that the decomposition of the tree can be computed in $\Theta(n\log n)$ time.
However, there exist construction algorithms to compute the decomposition in linear time \cite{BrodalFPO01,SPIRE19}. 

In the following, we will use the centroid decomposition \stcentroidT\ of the suffix tree \sufft\ of a string $S$. Given a node $u$ in \sufft , we use $\locus(u)$ to denote its locus, i.e., the string obtained by concatenating the sequence of labels encountered along the path from the root to $u$.

Assume we are given a pattern $P[1,p]$ and our goal is to find the suffix of $S$
that shares the longest common prefix with $P$. This problem can be easily solved with the suffix tree \sufft\ of $S$ with the following two-phase strategy: We first identify the highest node $u^*$ in \sufft\ such that $\locus(u^*)$ shares the longest common prefix with $P$.
Then, we try to extend the match by comparing the remaining characters of $P$ with the characters on the edge between $u^*$ and one of its children, i.e., the child where the label starts with the character $P[|\locus(u^*)|+1]$.

We can perform the same search for $u^*$ on the centroid decomposition \stcentroidT\ of the suffix tree of $S$. 
The search is done by traversing a root-to-node path of $O(\log n)$ nodes.
We start from the root of \stcentroidT\ and we move down to the leaves. 
For every node $u$ we visit, we compare $\locus(u)$ with $P$ and decide in which of its children we have to continue the search. As the target node $u^*$ is guaranteed to be  visited, we simply take track of the visited node sharing the longest common prefix with $P$. 
Based on the result of comparing $\locus(u)$ and $P$, there are the following cases:

\begin{itemize}
    \item If $\locus(u)$ equals $P$, then $u$ is our target node $u^*$ and we conclude.
    \item If $\locus(u)$ is not a prefix of $P$, we continue to search on the child of $u$ which corresponds to the connected component containing the parent of $u$ in the suffix tree. If such a node does not exist, we conclude. 
    \item If $\locus(u)$ is a prefix of $P$, 
    we continue the search on the connected component containing one of the children of $u$ in the suffix tree.
    The child is the node $v$ such that the first character of the edge between $u$ and $v$ equals the character $P[|\locus(u)|+1]$. This is exactly the node $v$ that a normal search on the suffix tree would visit next, once the search reaches $u$. Notice that in general $v$ is not a child of $u$ in the centroid decomposition. If such a node does not exist, we finish the visit.
\end{itemize}

Let us now show how to use the above algorithm to reconstruct an unknown string $S$.

\paragraph*{Solution with prefix queries to the oracle}
We first describe our algorithm for querying the oracle in an easier setting. 
Instead of answering substring queries, the oracle answers {\em prefix queries}:
given a string $P$, the oracle tells us whether $P$ is a prefix of the unknown string $S$.
This model is stronger because it allows us to remain anchored to the beginning of $S$ while 
reconstructing it.\footnote{An oracle for prefix queries can be obtained from an oracle for substring queries if we assume that $S$ begins with a special character $\$$ not belonging to $\Sigma$.} A direct consequence is that the algorithm is easier and faster. 

We now describe how to reconstruct a string $S$ with $\Theta(\sigma z_{no} \log n)$ prefix queries to the oracle. 

Our algorithm works in steps. In the $i$-th step it reconstructs the $i$-th LZ77 phrase 
$\lz{i}$. Once $\lz{i}$ is reconstructed, the algorithm knows the string $\rS{i}$, which is the  concatenation of all the phrases reconstructed so far. 
Observe that $\lz{i}$ is the longest substring of
$\rS{i-1}$ such that the prefix query $Q_i = \rS{i-1} \cdot \lz{i}$ is answered affirmatively.

The phrase $\lz{i}$ is identified with $O(\sigma(\log |\rS{i-1}| + 1))$ prefix queries as 
follows. Assume we have the suffix tree \sufft\ of $\rS{i-1}$ and its centroid decomposition \stcentroidT. 
Our first goal is to identify the lowest node $u^*$ in \sufft\ such that $\rS{i-1} \cdot 
\locus(u^*)$ is a prefix of $S$. This can be done by performing a search for the unknown 
pattern $P=\locus(u^*)$ on \stcentroidT. 
Even if $u^*$ is unknown, the search can be performed correctly. 
 Indeed, observe that $u^*$ and all its ancestors in \sufft\ are the only nodes $u$ 
such that $\rS{i-1} \cdot \locus(u)$ is a prefix of $S$. Thus, 
we perform the search on the centroid tree by binary searching for $u^*$ on root-to-$u^*$ path. 
The cost of the search is $O(\sigma(\log |\rS{i-1}| + 1))$ prefix queries.
Indeed, we need to visit $O(\log |\rS{i-1}| + 1)$ nodes of \stcentroidT\ to identify 
$u^*$. For each visited node $u$, we need a query to check if $\rS{i-1} \cdot \locus(u)$ is a prefix of $S$. If this is the case, at most $\sigma$ queries of the form $\rS{i-1} \cdot 
\locus(u) \cdot c$, with $c\in \Sigma$, are needed to know in which child of $u$ we have to continue our search. Otherwise, we move to the component
containing the parent of $u$, if any.

Once we know $u^*$, we have to extend $\locus(u^*)$ to match $\lz{i}$. 
Indeed, $\lz{i}$ may end up in the middle of the edge from $u^*$ to one of its children, say 
$v$, in \sufft. 
This step can be easily done with $O(\sigma + \log |\lz{i}|)$ queries. 
First, we use $O(\sigma)$ queries to identify the child $v$ of $u$, then we perform 
an exponential search on the length of the edge label. 

We conclude by proving that the reconstruction of $S$ takes $O(n\log n + n\log \sigma)$ time. 
A trivial implementation of our algorithm consists in rebuilding at each step $i$ the suffix tree of string $\rS{i}$ and its centroid decomposition from scratch. This takes quadratic time. 

A faster algorithm is the following:
First, we observe that, as the string is reconstructed from left to right, we can use the Ukkonen's construction of the suffix tree~\cite{Ukkonen95}. This construction builds the suffix tree in $O(n \log \sigma)$ time and linear space. It is an online algorithm that processes the string from left to right, hence it
allows us to build the suffix tree of the prefix of the string that we have already reconstructed.

The centroid decomposition of the suffix tree is instead kept updated dinamically. 
Brodal et al.~\cite{BrodalFPO01} showed how to keep an approximation of the centroid decomposition of a tree subject to insertions of new nodes in $O(\log n)$ amortized time per insertion. The decomposition is approximated in the sense that each selected centroid node splits the tree in connected components having a fraction $\frac{1}{2} + \epsilon$ of the overall tree dimension, for any $0 < \epsilon < \frac{1}{4}$. The height of the \centroidT\ is still $O(\log n)$, thus this approximated decomposition suffices for our purposes.

\paragraph*{Solution with substring queries to the oracle}
The string $S$ can be reconstructed with substring queries using an easy variant of the above algorithm. 
The algorithm reconstructs (a portion of) the string exactly as described above, until the string reconstructed so far, say $\rS{i}$, cannot be further extended to the right,  hence we know that 
$\rS{i}$ is a suffix of $S$. Then, we start extending $\rS{i}$ 
from its beginning, proceeding backwards (that is, prepending characters to $\rS{i}$).
More formally, our strategy first queries forward strings and builds the suffix tree and the LZ77 factorization of some suffix $S[i,n]$ of the string.  Then, we build the suffix tree of $\overleftarrow{S[i,n]}$ and proceed backwards, building the LZ77 factorization of the remaining portion $\overleftarrow{S[1,i-1]}$. 
Since the size $g$ of the smallest grammar is invariant under reversals and upper-bounds the number of Lempel-Ziv phrases, in both phases we generate at most $O(g)$ phrases. The following theorem is therefore immediate.

\begin{theorem}
Let $S$ be a string of length $n$ over the alphabet $\Sigma=[1,\ldots, \sigma]$. 
There exists an algorithm that reconstructs $S$ with $q = O(\sigma g\log n)$ substring queries to the oracle.  The algorithm runs in $O(n(\log n + \log \sigma)+ q)$ time using linear space.
\end{theorem}

Finally, Theorem~\ref{thm:lz} follows from the well-known bound $g\in O(z_{no}\log(n/z_{no}))$ (see also Navarro~\cite{navarro2020indexing}).

\paragraph*{Running example}

\begin{figure}
     \begin{subfigure}[b]{0.45\textwidth}
          \centering
          \resizebox{\linewidth}{!}{\input{suffixtree.tikz}}
          \caption{Suffix tree \sufft\ of $\rS{i-1}$}
          \label{fig:A}
     \end{subfigure}
     \begin{subfigure}[b]{0.55\textwidth}
          \centering
          \resizebox{\linewidth}{!}{\input{centroid.tikz}}
          \caption{Centroid decomposition \stcentroidT\ of \sufft }
          \label{fig:B}
     \end{subfigure}\\[1ex]
     \begin{subfigure}[b]{\textwidth}
          \centering
          \resizebox{\linewidth}{!}{\input{text.tikz}}
          \caption{The reconstructed string $\rS{i-1}$ and the unknown string $S$}
          \label{fig:C}
     \end{subfigure}
     \caption{A running example\label{fig:running}}
 \end{figure}
 
Suppose we have already reconstructed the string $\rS{i-1} = \tt{AAABCABCABCAAA}$.
This is a substring of the unknown string $S$ shown in Figure~\ref{fig:C}. There are two occurrences of $\rS{i-1} $ in $S$ and the red cells highlight the characters that we still need to learn.
The suffix tree \sufft\ of $\rS{i-1}$ (see Figure~\ref{fig:A}) has been built online with Ukkonen's algorithm and it will be updated as soon as we learn more characters. For this reason we do not append any special character at the end of $\rS{i-1}$.
Thus, there may exist suffixes of $\rS{i-1}$ which do not have their leaves in \sufft\ because they are proper
prefixes of some another suffix. In our example this happens to the last three suffixes $\tt{A}$, $\tt{AA}$ and $\tt{AAA}$ which are proper prefixes of the whole string $\rS{i-1}$.
Nodes of \sufft\ are numbered (in our example levelwise just for convenience) to map the corresponding node in the centroid decomposition \stcentroidT (see Figure~\ref{fig:B}).

The label on the edge from node $u$ to its child $v$ reports the interval $[i,j]$ of positions on string $\rS{i-1}$ representing the locus of node $v$. For example, the label on the edge from node $5$ to node $13$ is $[3,9]$ and, thus, $\locus(13)={\tt ABCABCA}$.
In the centroid decomposition \stcentroidT, each node $u$ is labeled with the interval of positions of
$\locus(u)$. The label of node $13$ is $[3,9]$ because $\locus(13)={\tt ABCABCA}$.
In the centroid tree, the leftmost child of any node $u$ is the centroid decomposition of the connected component containing the parent of 
$u$ (if any) while the other children are the centroid decompositions of the subtrees rooted at the children of $u$ in \sufft (if any).
Note that for any child $v$ of node $u$ but, possibly, the leftmost one, we have that $\locus(u)$ is a prefix of $\locus(v)$.

We start from the root of \stcentroidT\ (node $0$) which in our example, by coincidence, corresponds to the root of \sufft . We query the oracle for the substring $\rS{i-1}\cdot \locus(0)$. As $\locus(0)$ is the empty string, the oracle's answer will be positive. The algorithm continues on a child of node $0$. 
 
For any child $v$ of $u$, let be $c_v$ the character such that $\locus(u)\cdot c_v$ is a prefix of $\locus(v)$, 
i.e., $c_v = \locus(v)[|\locus(u)|+1]$. 
We process the children of node $0$ and continue on a node $v$ such that the query for the substring $\rS{i-1}\cdot \locus(0)\cdot c_v$
is successful. 
There may be several such nodes $v$.
For example, we could continue on both nodes $5$ and $7$. This is because both substrings $\rS{i-1}\cdot {\tt A}$ and $\rS{i-1}\cdot {\tt B}$ occur in $S$. We can arbitrarily choose any of these nodes but, of course, the length of the substring we reconstruct may vary.
Suppose we continue on node $5$. Then, we 
ask for the substring  $\rS{i-1} \cdot \locus(5)$. 
As the answer is positive, we continue with node $12$
because $c_{12}= {\tt B}$ and $\rS{i-1} \cdot \locus(5)\cdot c_{12}$ occurs in $S$.
We query for $\rS{i-1} \cdot \locus(12)$.
As the answer is negative, we binary search 
for the longest prefix $P$ of $\locus(12)$ such that 
$\rS{i-1}\cdot P$ occurs in $S$. 
This way, we reconstruct the substring $X ={\tt ABCAB}$ and we learn $\rS{i} = \rS{i-1}\cdot X$.

\section{Conclusions and future work}\label{sec:final}

We investigated the connection between the string reconstruction and compression problems, establishing that they essentially coincide: the number of substring queries that need to be asked to an oracle in order to reconstruct a string $S$ is proportional to the complexity of $S$.

We also showed that it is possible to efficiently reconstruct a string of length $n$ over an alphabet of size $\sigma$ using $O(\sigma g \log n) \subseteq O(\sigma \cdot z_{no} \log(n/z_{no})\log n)$ queries in $O(n(\log \sigma + \log n))$ time, where $z_{no}$ is the number of phrases of the LZ77 factorization of $S$  without overlaps and $g$ is the size of the smallest grammar producing $S$. Immediate improvements over our work would be to replace $z_{no}$ with the more powerful $z$ (i.e., allowing overlaps), or to shave log factors from the complexities of our algorithms.
In particular, we know that 
the number of queries cannot be improved by more than a factor $O(\log \frac{n}{z_{no}} \log \sigma)$ in general.

In our setting, we aim at reconstructing the whole unknown string $S$. One can also consider the problem of reconstructing the set of all substrings of $S$ of a given length $k$, see for example \cite{Tsur05}. Notice that knowing all the substrings of $S$ of length $r(S)+2$ allows one to uniquely determine $S$, where $r(S)$ is the repetition index of $S$, that is, the length of the longest repeat of $S$ \cite{CaDel01,Fi06}. 

Another direction of investigation consists in introducing uncertainty into the model. For example, allowing the oracle to answer the queries with a certain probability of returning a wrong result --- this could model strings with character ambiguities, e.g., DNA strings arising from a sequencing --- or allowing the oracle to return positive answers to queries within a limited Hamming distance from substrings of the target string.

%

 \bibliographystyle{elsarticle-num}





\end{document}

%% file: suffixtree.tikz
\begin{tikzpicture}[
level distance=15mm,
level 1/.style={sibling distance=40mm},
level 2/.style={sibling distance=30mm},
level 3/.style={sibling distance=20mm},
level 4/.style={sibling distance=20mm},
]

\tikzstyle{c} = [draw, shape=circle, minimum width=9mm]
\tikzstyle{label} = [draw=white, shape=rectangle,minimum width=10mm, fill=white, font=\footnotesize]

\node[c]{0}[edge from parent]
    child {node[c] {1}
        child {node[c]{4}
          child {node[c]{10}
            edge from parent coordinate (e10);
            }
          child {node[c]{11}
            edge from parent coordinate (e11);
            }
          edge from parent coordinate (e4);
          }
        child {node[c]{5}
          child {node[c]{12}
            edge from parent coordinate (e12);
            }
          child {node[c]{13}
            child {node[c]{18}
              edge from parent coordinate (e18);
              }
            child {node[c]{19}
              edge from parent coordinate (e19);
              }
            edge from parent coordinate (e13);
            }
            edge from parent coordinate (e5);
          }
        edge from parent coordinate (e1);
    }
    child {node[c] {2}
        child {node[c]{6}
          edge from parent coordinate (e6);
          }
        child {node[c]{7}
          child {node[c]{14}
            edge from parent coordinate (e14);
            }
          child {node[c]{15}
            edge from parent coordinate (e15);
            }
          edge from parent coordinate (e7);
          }
        edge from parent coordinate (e2);
    }
    child {node[c] {3}
        child {node[c]{8}
          edge from parent coordinate (e8);
        }
        child {node[c]{9}
          child {node[c]{16}
            edge from parent coordinate (e16);
            }
          child {node[c]{17}
            edge from parent coordinate (e17);
            }
          edge from parent coordinate (e9);
        }
        edge from parent coordinate (e3);
    };
\node[label] at(e1)  {$[1,1]$};
\node[label] at(e2) {$[4,6]$};
\node[label] at(e3) {$[5,6]$};
\node[label] at(e4) {$[1,2]$};
\node[label] at(e5) {$[3,6]$};
\node[label] at(e6) {$[10,14]$};
\node[label] at(e7) {$[4,9]$};
\node[label] at(e8) {$[11,14]$};
\node[label] at(e9) {$[5,9]$};
\node[label] at(e10) {$[1,14]$};
\node[label] at(e11) {$[2,14]$};
\node[label] at(e12) {$[9,14]$};
\node[label] at(e13) {$[3,9]$};
\node[label] at(e14) {$[7,14]$};
\node[label] at(e15) {$[4,14]$};
\node[label] at(e16) {$[8,14]$};
\node[label] at(e17) {$[5,14]$};
\node[label] at(e18) {$[6,14]$};
\node[label] at(e19) {$[3,14]$};
\end{tikzpicture}

%% file: centroid.tikz
\begin{tikzpicture}[
level distance=15mm,
level 1/.style={sibling distance=60mm},
level 2/.style={sibling distance=20mm},
level 3/.style={sibling distance=20mm},
level 4/.style={sibling distance=20mm},
]

\tikzstyle{c} = [draw, shape=circle, minimum width=9mm]
\tikzstyle{lab} = [font=\footnotesize]

\node[c] (n0) {0}[edge from parent]
    child {node[c] (n5) {5}
      child {node[c] (n4) {4}
          child {node[c] (n1) {1}
              edge from parent coordinate (e1);
          }
          child {node[c] (n10) {10}
              edge from parent coordinate (e10);
          }
          child {node[c] (n11) {11}
              edge from parent coordinate (e11);
          }
          edge from parent coordinate (e4);
          }
      child {node[c] (n12) {12}
          edge from parent coordinate (e12);
          }
      child {node[c] (n13) {13}
          child {node[c] (n18) {18}
              edge from parent coordinate (e18);
              }
          child {node[c] (n19) {19}
              edge from parent coordinate (e19);
              }
          edge from parent coordinate (e13);
          }
      edge from parent coordinate (e5);
    }
    child {node[c] (n7) {7}
        child {node[c] (n2) {2}
            child {node[c] (n6) {6}
              edge from parent coordinate (e6);
            }
          edge from parent coordinate (e2);
        }
        child {node[c] (n14) {14}
          edge from parent coordinate (e14);
        }
        child {node[c] (n15) {15}
          edge from parent coordinate (e15);
        }
      edge from parent coordinate (e7);
    }
    child {node[c] (n9) {9}
        child {node[c] (n3) {3}
            child {node[c] (n8) {8}
              edge from parent coordinate (e8);
            }
          edge from parent coordinate (e3);
        }
        child {node[c] (n16) {16}
          edge from parent coordinate (e16);
        }
        child {node[c] (n17) {17}
          edge from parent coordinate (e17);
        }
      edge from parent coordinate (e9);
    };
\node[lab, right of=n0, node distance=0.85cm] {$[0,0]$};
\node[lab, below of=n1, node distance=0.7cm] {$[1,1]$};
\node[lab, right of=n2, node distance=0.85cm] {$[4,6]$};
\node[lab, right of=n3, node distance=0.85cm] {$[5,6]$};
\node[lab, right of=n4, node distance=0.85cm] {$[1,2]$};
\node[lab, right of=n5, node distance=0.85cm] {$[3,6]$};
\node[lab, below of=n6, node distance=0.7cm] {$[10,14]$};
\node[lab, right of=n7, node distance=0.85cm] {$[4,9]$};
\node[lab, below of=n8, node distance=0.7cm] {$[11,14]$};
\node[lab, right of=n9, node distance=0.85cm] {$[5,9]$};
\node[lab, below of=n10, node distance=0.7cm] {$[1,14]$};
\node[lab, below of=n11, node distance=0.7cm] {$[2,14]$};
\node[lab, right of=n12, node distance=0.85cm] {$[3,9]$};
\node[lab, right of=n14, node distance=0.85cm] {$[7,14]$};
\node[lab, right of=n13, node distance=0.85cm] {$[3,9]$};
\node[lab, right of=n15, node distance=0.85cm] {$[4,14]$};
\node[lab, right of=n16, node distance=0.85cm] {$[8,14]$};
\node[lab, right of=n17, node distance=0.85cm] {$[5,14]$};
\node[lab, below of=n18, node distance=0.7cm] {$[6,14]$};
\node[lab, below of=n19, node distance=0.7cm] {$[3,14]$};
\end{tikzpicture}

%% file: text.tikz
\begin{tikzpicture}
\node [left] at (0,0) {$\rS{i-1}$};
\foreach \x in {1,...,14}
  \node at (\x*0.5, 0.4) {{\scriptsize \x}};
\foreach \x\v in {1/A,2/A,3/A,4/B,5/C,6/A,7/B,8/C,9/A,10/B,11/C,12/A,13/A,14/A}
  \node [draw, fill=lightgray,minimum width=0.5cm] at (\x*0.5, 0) {\v};

\node [left] at (0,-1) {$S$};
\foreach \x\v in {1/A,2/A,3/A,4/B,5/C,6/A,7/B,8/C,9/A,10/B,11/C,12/A,13/A,14/A,15/B,16/C,17/A,18/B,19/C,20/A}
    \node [draw, fill=lightgray, minimum width=0.5cm] at (\x*0.5, -1) {\v};

\foreach \x\v in {15/B,16/C,17/A,18/B,19/C,20/A}
    \node [draw, fill=red,minimum width=0.5cm] at (\x*0.5, -1) {\v};

\foreach \x\v in {21/A,22/A,23/A,24/B,25/C,26/A,27/B,28/C,29/A,30/B,31/C,32/A,33/A,34/A}
    \node [draw, fill=lightgray, minimum width=0.5cm] at (\x*0.5, -1) {\v};

\foreach \x\v in {35/A,36/B,37/C,38/A,39/B}
    \node [draw, fill=red,minimum width=0.5cm] at (\x*0.5, -1) {\v};

\end{tikzpicture}

%% file: mainTCS.bbl
\begin{thebibliography}{10}
\expandafter\ifx\csname url\endcsname\relax
  \def\url#1{\texttt{#1}}\fi
\expandafter\ifx\csname urlprefix\endcsname\relax\def\urlprefix{URL }\fi
\expandafter\ifx\csname href\endcsname\relax
  \def\href#1#2{#2} \def\path#1{#1}\fi

\bibitem{JiangLi94}
T.~Jiang, M.~Li, {DNA} sequencing and string learning, Mathematical systems
  theory 29~(4) (1996) 387--405.
\newblock \href {https://doi.org/10.1007/BF01192694}
  {\path{doi:10.1007/BF01192694}}.

\bibitem{MS95}
D.~Margaritis, S.~Skiena, Reconstructing strings from substrings in rounds, in:
  36th Annual Symposium on Foundations of Computer Science, Milwaukee,
  Wisconsin, USA, 23-25 October 1995, {IEEE} Computer Society, 1995, pp.
  613--620.
\newblock \href {https://doi.org/10.1109/SFCS.1995.492591}
  {\path{doi:10.1109/SFCS.1995.492591}}.

\bibitem{DBLP:journals/jcb/SkienaS95}
S.~Skiena, G.~Sundaram, Reconstructing strings from substrings, J. Comput.
  Biol. 2~(2) (1995) 333--353.
\newblock \href {https://doi.org/10.1089/cmb.1995.2.333}
  {\path{doi:10.1089/cmb.1995.2.333}}.

\bibitem{Tsur05}
D.~Tsur, Tight bounds for string reconstruction using substring queries, in:
  C.~Chekuri, K.~Jansen, J.~D.~P. Rolim, L.~Trevisan (Eds.), Approximation,
  Randomization and Combinatorial Optimization, Algorithms and Techniques, 8th
  International Workshop on Approximation Algorithms for Combinatorial
  Optimization Problems, {APPROX} 2005 and 9th International Workshop on
  Randomization and Computation, {RANDOM} 2005, Berkeley, CA, USA, August
  22-24, 2005, Proceedings, Vol. 3624 of Lecture Notes in Computer Science,
  Springer, 2005, pp. 448--459.
\newblock \href {https://doi.org/10.1007/11538462\_38}
  {\path{doi:10.1007/11538462\_38}}.

\bibitem{DBLP:journals/siamdm/AcharyaDMOP15}
J.~Acharya, H.~Das, O.~Milenkovic, A.~Orlitsky, S.~Pan, String reconstruction
  from substring compositions, {SIAM} J. Discret. Math. 29~(3) (2015)
  1340--1371.
\newblock \href {https://doi.org/10.1137/140962486}
  {\path{doi:10.1137/140962486}}.

\bibitem{Dress-Erdos}
A.~W.~M. Dress, P.~L. Erd{\H o}s, Reconstructing words from subwords in linear
  time, Annals of Combinatorics 8~(4) (2005) 457--462.
\newblock \href {https://doi.org/10.1007/s00026-004-0232-4}
  {\path{doi:10.1007/s00026-004-0232-4}}.

\bibitem{DBLP:conf/spire/AfsharAGM20}
R.~Afshar, A.~Amir, M.~T. Goodrich, P.~Matias, Adaptive exact learning in a
  mixed-up world: Dealing with periodicity, errors and jumbled-index queries in
  string reconstruction, in: {SPIRE} 2020: Proceedings of the 27th
  International Symposium on String Processing and Information Retrieval, Vol.
  12303 of Lecture Notes in Computer Science, Springer, 2020, pp. 155--174.
\newblock \href {https://doi.org/10.1007/978-3-030-59212-7\_12}
  {\path{doi:10.1007/978-3-030-59212-7\_12}}.

\bibitem{DBLP:journals/corr/abs-1808-00674}
K.~Iwama, J.~Teruyama, S.~Tsuyama,
  \href{http://arxiv.org/abs/1808.00674}{Reconstructing strings from
  substrings: Optimal randomized and average-case algorithms}, CoRR
  abs/1808.00674 (2018).
\newline\urlprefix\url{http://arxiv.org/abs/1808.00674}

\bibitem{LZ77}
J.~Ziv, A.~Lempel, A universal algorithm for sequential data compression, IEEE
  Transactions on Information Theory 23~(3) (1977) 337--343.
\newblock \href {https://doi.org/10.1109/TIT.1977.1055714}
  {\path{doi:10.1109/TIT.1977.1055714}}.

\bibitem{navarro2020indexing}
G.~Navarro, Indexing highly repetitive string collections (2020).
\newblock \href {http://arxiv.org/abs/2004.02781} {\path{arXiv:2004.02781}}.

\bibitem{naor91}
M.~Naor, String matching with preprocessing of text and pattern, in: ICALP
  1991: Proceedings of the 18th International Colloquium on Automata,
  Languages, and Programming, Springer Berlin Heidelberg, Berlin, Heidelberg,
  1991, pp. 739--750.

\bibitem{BenderFK06}
M.~A. Bender, M.~Farach{-}Colton, B.~C. Kuszmaul, Cache-oblivious string
  b-trees, in: PODS 2006: Proceedings of the Twenty-Fifth {ACM}
  {SIGACT-SIGMOD-SIGART} Symposium on Principles of Database Systems, 2006, pp.
  233--242.
\newblock \href {https://doi.org/10.1145/1142351.1142385}
  {\path{doi:10.1145/1142351.1142385}}.

\bibitem{FerraginaV16}
P.~Ferragina, R.~Venturini, Compressed cache-oblivious string {B}-tree, ACM
  Transactions on Algorithms (TALG) 12~(4) (2016) 52:1--52:17.
\newblock \href {https://doi.org/10.1145/2903141} {\path{doi:10.1145/2903141}}.

\bibitem{amir92}
A.~Amir, M.~Farach, Y.~Matias, Efficient randomized dictionary matching
  algorithms, in: Combinatorial Pattern Matching, Springer Berlin Heidelberg,
  Berlin, Heidelberg, 1992, pp. 262--275.

\bibitem{Jordan1869}
C.~Jordan, Sur les assemblages de lignes, Journal für die reine und angewandte
  Mathematik 70 (1869) 185--190.

\bibitem{BrodalFPO01}
G.~S. Brodal, R.~Fagerberg, C.~N.~S. Pedersen, A.~{\"{O}}stlin, The complexity
  of constructing evolutionary trees using experiments, in: ICALP 2001:
  Proceedings of the 28th International Colloquium on Automata, Languages and
  Programming, 2001, pp. 140--151.
\newblock \href {https://doi.org/10.1007/3-540-48224-5\_12}
  {\path{doi:10.1007/3-540-48224-5\_12}}.

\bibitem{SPIRE19}
D.~{Della Giustina}, N.~Prezza, R.~Venturini, A new linear-time algorithm for
  centroid decomposition, in: {SPIRE} 2019: {P}roceedings of the 26th
  {I}nternational {S}ymposium on {S}tring {P}rocessing and {I}nformation
  {R}etrieval, Springer, 2019, pp. 274--282.
\newblock \href {https://doi.org/10.1007/978-3-030-32686-9\_20}
  {\path{doi:10.1007/978-3-030-32686-9\_20}}.

\bibitem{Ukkonen95}
E.~Ukkonen, On-line construction of suffix trees, Algorithmica 14~(3) (1995)
  249--260.
\newblock \href {https://doi.org/10.1007/BF01206331}
  {\path{doi:10.1007/BF01206331}}.

\bibitem{CaDel01}
A.~Carpi, A.~de~Luca, Words and special factors, Theoret. Comput. Sci.
  259~(1-2) (2001) 145--182.
\newblock \href {https://doi.org/10.1016/S0304-3975(99)00334-5}
  {\path{doi:10.1016/S0304-3975(99)00334-5}}.

\bibitem{Fi06}
G.~Fici, F.~Mignosi, A.~Restivo, M.~Sciortino, Word assembly through minimal
  forbidden words, Theor. Comput. Sci. 359~(1-3) (2006) 214--230.
\newblock \href {https://doi.org/10.1016/j.tcs.2006.03.006}
  {\path{doi:10.1016/j.tcs.2006.03.006}}.

\end{thebibliography}
